\documentclass[english]{article}
\usepackage{geometry}
\usepackage{float}
\usepackage{mathrsfs}
\usepackage{amsmath}
\usepackage{amssymb}
\usepackage{graphicx}
\usepackage{amsfonts}
\usepackage{caption}
\usepackage{subcaption}
\usepackage[]{algorithm2e}
\usepackage{amsthm}
\usepackage{epsfig}
\usepackage{comment}
\usepackage{babel}\date{}
\usepackage{subfig}\usepackage[all]{xy}
\usepackage{color}
\makeatletter

\definecolor{darkred}{rgb}{.7,0,0}
\definecolor{darkgreen}{rgb}{0,0.5,0}
\definecolor{darkblue}{rgb}{0,0,0.7}
\definecolor{darkyellow}{rgb}{0.0,0.704,0.7}

\newcommand{\km}{k_{\rm min}}

\@ifundefined{definecolor}{\@ifundefined{definecolor}
 {\@ifundefined{definecolor}
 {\usepackage{color}}{}
}{}
}{}

\newtheorem{theorem}{Theorem}[section]

\newtheorem{lem}{Lemma}[section]
\newtheorem{rem}{Remark}[section]

\newcounter{hypA}

\newcounter{hypB}

\def\cP{{\mathsf P}}

\def\*{|\!|\!|}

\newcommand{\bbR}{\mathbb R}

\newcommand{\E}{\mathbb{E}}

\renewcommand{\phi}{\varphi}

\newcommand{\bbE}{\E}

\newcommand{\bF}{\overline{f}}

\makeatother

\begin{document}

\begin{center}

{\Large \textbf{Sequential Monte Carlo Methods for Bayesian Elliptic Inverse Problems}}

\bigskip

BY ALEXANDROS BESKOS$^{1}$, AJAY JASRA$^{2}$, EGE A. MUZAFFER$^{2}$ \& ANDREW M. STUART$^{3}$

{\footnotesize $^{1}$Department of Statistical Science,
University College London, London, WC1E 7HB, UK.}\\
{\footnotesize E-Mail:\,}\texttt{\emph{\footnotesize alex@stats.ucl.ac.uk}}\\
{\footnotesize $^{2}$Department of Statistics \& Applied Probability,
National University of Singapore, Singapore, 117546, SG.}\\
{\footnotesize E-Mail:\,}\texttt{\emph{\footnotesize staja@nus.edu.sg, m.ege85@nus.edu.sg}}\\
{\footnotesize $^{3}$Department of Mathematics,
University of Warwick, Coventry, CV4 7AL, UK.}\\
{\footnotesize E-Mail:\,}\texttt{\emph{\footnotesize a.m.stuart@warwick.ac.uk}}
\end{center}

\begin{abstract}
In this article we consider a Bayesian inverse problem associated to elliptic partial differential equations (PDEs) in two
and three dimensions. This class of inverse problems is important in 
applications such as hydrology, but the complexity of the link function
between unknown field and measurements can make it difficult to 
draw inference from the associated posterior. We prove that for this inverse
problem a basic SMC method has a Monte Carlo rate of
convergence with constants which are independent of the dimension of
the discretization of the problem; indeed convergence of the SMC method
is established in a function space setting. We also develop an enhancement 
of the sequential Monte Carlo (SMC) methods for inverse problems which were
introduced in \cite{kantas}; the enhancement is designed to deal with 
the additional complexity of this elliptic inverse problem.  The efficacy 
of the methodology, and its desirable theoretical properties, are
demonstrated on numerical examples in both two and three dimensions.
\\
\textbf{Keywords}: Inverse Problems, Elliptic PDEs, Groundwater
Flow, Adaptive SMC, Markov chain Monte Carlo.
\end{abstract}

\section{Introduction}\label{sec:intro}

The viability of the Bayesian approach to inverse problems was
established in the pioneering text \cite{KS05} which, in particular,
demonstrated the potential for Markov chain Monte Carlo (MCMC) methods in this context. Nonetheless,
the high dimensional nature of the unknown, often found from
discretizing a field, leads to difficult problems in the design of
proposals which are cheap to implement, yet which mix efficiently. 
One recent approach to tackle these problems has been the
development of algorithms with mesh-free mixing times, such as those
highlighted in \cite{cotter,hoang}; these non-standard MCMC algorithms 
avoid the unnecessary penalties incurred by naive proposals related
to exploration of the part of the parameter space dominated by the prior. 
Nonetheless, in the large dataset or small observational noise
regimes, one is still confronted with an inference problem 
in high dimensions which, whilst of smaller order than the dimension 
of the PDE solver, exhibits wide variations in scales in different
coordinates of the parameterizations, 
leading to substantial challenges for algorithmic tuning.

A different approach, which we will adopt here, involves SMC samplers \cite{delm:06}.
These are particle methods which, in the context of Bayesian inverse problems,
build an approximation to a sequence of measures which interpolate
from the prior to the posterior; the sequential nature of the approximation
allows for adaptation of the particle distribution and weights from
the (typically simple) prior to the (potentially very complex) posterior.
Recent work in the context of inverse problems \cite{kantas} has shown how, by
using the aforementioned dimension-independent MCMC methods within SMC,
it is possible to construct algorithms which combine the desirable
dimension-independent aspects of novel MCMC algorithms with the
desirable self-adaptation of particle methods. This combination is beneficial
for complex posteriors such as those arising in the large dataset or
small noise regimes; in particular the computational results in
\cite{kantas} demonstrate an order of magnitude speed-up of these new
SMC methods over the MCMC methods highlighted in \cite{cotter}, 
within the context of the inverse problem for the initial condition 
of the Navier-Stokes equation. Furthermore, 
recent works \cite{beskos,delmoral,delm:06,jasra} have shown that important aspects of this SMC algorithm for inverse problems, such as adaptation, tempering 
and parallelization, have the potential to provide effective methods even 
for high-dimensional inverse problems.


The contributions of this article are three-fold:
\begin{enumerate}
\item{A computational study of SMC methods for
a class of Bayesian inverse problems which arise in
applications such as hydrology \cite{iglesias}, 
and are more challenging to fit, in comparison 
to the Navier-Stokes inverse problem which was the focus of
the development of novel SMC methods in \cite{kantas}; furthermore, with
modification of the measurement set-up, the inverse problems considered
also find application in medical imaging problems such as EIT \cite{KS05}.}

\item{An enhancement of the class of SMC methods introduced in 
\cite{kantas} which leads to greater efficiency and, in particular, the 
ability to efficiently solve the elliptic inverse problems which are the 
focus of this paper.}

\item{A proof of the fact that these SMC algorithms have Monte Carlo
convergence rates which are mesh-independent and, indeed converge in
the function space setting. This complements related theoretical work
\cite{HSV} which establishes mesh-independence mixing for the novel
MCMC methods which are used as proposal kernels within the SMC
approach of \cite{kantas} which we build upon here.}
\end{enumerate}

This article is structured as follows. In Section \ref{3dinv} we describe the Bayesian model and associated PDE. In Section \ref{sec:SMC} our computational procedure is outlined, along
with our theoretical results. In Section \ref{sec:numerics} we present our numerical results. The article is concluded in Section \ref{sec:summ} with a discussion of areas for future work.

\section{Modelling}
\label{3dinv}
Consider two normed linear spaces, $\mathcal{K}$ and $\mathcal{Y}\subseteq\mathbb{R}$, corresponding to the state space of the parameters ($u\in\mathcal{K}$) and observations
($y\in\mathcal{Y}$) respectively. We will observe data at spatial locations $x\in\mathcal{X}\subset\mathbb{R}^d, d\in\{2,3\}$ and we denote the observation at location $x$ as $y_x$.
Let $G:\mathcal{X}\times \mathcal{K}\rightarrow\mathcal{Y}$ 
and, for each $x \in \mathcal{X}$, let $\epsilon_x\in\mathcal{Y}$ be a random variable of zero mean; 
then we will be concerned with models of the form:
$$
y_x = G(x;u) + \epsilon_x.
$$
Here $G(x;u)$ is an underlying system behaviour for a given parameter $u$,
and  $\epsilon_x$ expresses measurement (and indeed sometimes model error) 
at location $x$. 
In our context, $G$ is associated to the solution of a PDE, with parameter $u$. We are interested in drawing inference on $u$, given a prior distribution on $u\in\mathcal{K}$, conditional upon observing realizations of $y_x$ for 
a set of points $x\in O\subseteq\mathcal{X}$, with $\textrm{Card}(O)<+\infty$. This is the framework of our Bayesian inverse problem.
In subsection \ref{ssec:fm} we define the forward model, and in
subsection \ref{ssec:modelK} we describe prior modelling on our unknown.
Subsection \ref{ssec:bayesian} shows that the posterior distribution is
well-defined and states a key property of the log-likelihood, used in
what follows.

\subsection{Forward Model}
\label{ssec:fm}

In this paper, we focus on the general scenario where the forward map $G$ is described by an elliptic PDE. In particular, we work with a 
problem of central significance in hydrology, namely the 
estimation of subsurface flow from measurements of the pressure 
(hydraulic head) at certain locations $x$ in the domain of interest. The pressure
and velocity are linked by Darcy's law in which the subsurface permeability
appears as a parameter; estimating it is thus a key step in predicting the
subsurface flow. In this subsection we define the forward map from
permability to pressure space.

In detail, we consider the $d-$dimensional cube 
$ \mathcal{X}= [ \frac{-\pi}{2}, \frac{\pi}{2} ]^d$ as our domain, in
both the cases $d=2,3.$
Define a mapping $p: \mathcal{X} \to \mathbb{R}$, denoting pressure
(or hydraulic head), $v:\mathcal{X} \to \mathbb{R}^3$, denoting a quantity
proportional to velocity, and $u:\mathcal{X} \to \mathbb{R}$, denoting permeability (or hydraulic conductivity) of soil \cite{MT}.
The behaviour of the system is described through the elliptic PDE:
\begin{subequations}
\label{eq2}
\begin{align}
v(x)&=-u(x) \nabla_x p(x),&x\in \mathcal{X}&\\
 -\nabla_x \cdot \bigl(v(x)\bigr) &= f(x),&x\in \mathcal{X}&\\
 p(x) & = 0,&x\in \partial  \mathcal{X}&.
\end{align}
\end{subequations}
Equation (\ref{eq2}a) is Darcy's law and contains the 
permeability $u$ (for us the key parameter); 
equation (\ref{eq2}b) expresses continuity of
mass and here $f:\mathcal{X} \to \mathbb{R}$ is assumed known and 
characterizes the source/sink configuration; in equation (\ref{eq2}c) 
$\partial  \mathcal{X}$ is the boundary of the domain, 
and we are thus assuming a homogeneous boundary condition on the boundary
pressure -- other boundary conditions, specifying the flux, are also possible.
Together equations \eqref{eq2} define an elliptic PDE for pressure $p$.

If $u$ is in $L^{\infty}(\mathcal{X})$ and lower bounded by a positive
constant $k_{\rm min}$ a.e. in $\mathcal{X}$
then, for every $f \in H^{-1}(\mathcal{X})$, there is a unique solution 
$p \in H^1_0(\mathcal{X})$ to the PDE \eqref{eq2} satisfying
\begin{equation}
\label{eq:est}
\|p\|_{H^1_0} \le \frac{1}{k_{\rm min}}\|f\|_{H^{-1}};
\end{equation}
see \cite{dashti} and the references therein.
In this setting, the forward map $G(x;u):= p(x)$ is well-defined
and thus corresponds to 
solution of the elliptic PDE for a given permeability field $u$. 
A typical choice of the source/sink function $f$ is
\begin{equation}
f(x) = \sum_i c_i \delta_{x_i}(x).  
\label{eq3}
\end{equation}
The set of points $\{x_i\}$ denote  the known position of sources or sinks, 
and the signs of each $c_i$ determine whether one has a source or sink
at $x_i$ \cite{iglesias}. We note that the cleanest setting for the
mathematical formulation of the problem requires $f \in H^{-1}(\mathcal{X})$
and, in theory, will require mollification of the Dirac's at each $x_i$;
in practice this modification makes little difference to the inference.

\subsection{Prior Modelling of Permeablity $u$}
\label{ssec:modelK}

We describe the modelling of $u$ in three dimensions; simplification
to the two dimensional setting is straightforward.
We begin by expressing the unknown model parameter as a Fourier series:
\begin{equation}
u(x) = \bar{u}(x) + \sum_{k\in\mathbb{Z}^{3}} u_k e_k(x). 
\label{parameter_expansion}
\end{equation}
Here we have scaled Fourier coefficients
\begin{equation}
 e_k(x) = a_k \exp(i k \cdot x)\ ,\quad k\in \mathbb{Z}^3\ 
\label{basis}
\end{equation}
and the real coefficients $\{a_k\}$,
complex coefficients $\{u_k\}$ (satisfying $\bar{u}_k = u_{-k}$)
and real-valued $L^{\infty}({\mathcal X})$ function $\bar{u}$ will be
chosen to enforce the mathematically necessary (and physically 
sensible) positivity restriction 
\begin{equation}
\label{eq:deen}
u(x)\ge \km>0,  \quad\quad x \in \mathcal{X}.
\end{equation} 
The use of Fourier series in principle enables the representation of
arbitrary functions in $L^2(\mathcal{X})$ by use of periodic
extensions. However we will impose a rate of decay on the $\{a_k\}$,
in order to work in the setting of inversion for this problem,
as developed in \cite{hoang,schwab}; this rate of decay will imply
a certain degree of smoothness in the function $(u-\bar{u})(\cdot).$
Noting that the functions $\exp(i k \cdot x)$ have $L^{\infty}(\mathcal{X})$
norm equal to one, we can place ourselves in the setting of \cite{hoang,schwab}
by assuming that, for some $q>0$, $C>0$,
\begin{equation}
\label{eq:cond}
\sum_k | a_k |_{\infty} < \infty\ , \quad \sum_{k:|k|_{\infty}>j} | a_k |_{\infty} < C j^{-q}.
\end{equation}
We choose $a_k$ of the form 
\begin{equation}
a_k=a |k|_{L^{\infty}}^{-\alpha}\label{eq:adef}
\end{equation}
and then impose $\alpha > 3$ in dimension $d=3$ or
$\alpha>2$ in dimension $d=2$.

Given this set-up, we need to find a suitable prior for $u$, 
so that the forward model $G(x;u)$ is almost-surely well-defined, as well 
as reflecting any prior statistical information we may have.
There are several widely adopted approaches in the literature for prior 
parameterization of the permeability, the most common being the 
log-normal choice (see \cite{tarantola} for details and, for example, \cite{iglesias}
for a recent application),
widely adopted by geophysicists, and the uniform case \cite{hoang,schwab}
which has been succesfully adopted in the computational mathematics
literature, building on earlier work of Schwab in uncertainty
quantification \cite{GS}. We work with the uniform priors popularized
by Schwab: we choose $u_k\stackrel{i.i.d.}{\sim}\mathcal{U}_{[-1,1]}$
in the representation  \eqref{parameter_expansion} for $u$, resulting
in a pushforward measure $\nu_0$ on $u$ as in \cite{schwab}. We let
$\mathcal{K}$ denote the separable 
Banach space found from the closure, with respect
to the $L^{\infty}$ norm, of the set of functions used in the representation 
\eqref{parameter_expansion} of $u(\cdot)$. Then $\nu_0$ is viewed as a measure
on $\mathcal{K}$; see \cite{dashti} for further details. Once the
parameters $a_k$ are chosen to satisfy \eqref{eq:cond}, 
the mean function ${\bar u}(x)$ can be chosen to ensure that there is $\km$ 
such that $u(\cdot)$ is in $L^{\infty}(\mathcal{X})$ and satisfies
\eqref{eq:deen} almost surely with respect to the prior $\nu_0$ on function
$u$.

\subsection{Bayesian Inverse Problem}
\label{ssec:bayesian}

We observe the pressure at certain locations, 
the set of which is denoted as $O\in\mathcal{X}$.
We will suppose that for each $x\in O$ and independently, 
$\epsilon_x\sim\mathcal{N}(0,\sigma^2)$, where $\mathcal{N}(0,\sigma^2)$ is the normal distribution of mean 0 and known variance $\sigma^2$.
Then the log-likelihood is, up to an irrelevant additive
constant, given by
\begin{equation} 
\Phi(u;y):=- \sum_{x \in O}{ \frac{\bigl|G(x;u) - y(x)\bigr|^2}{2\sigma^2} }.
\label{eq1}
\end{equation}

Along with the prior modelling in subsection \ref{ssec:modelK}, this defines a scenario so that the forward model $G(x;\cdot)$ is, almost-surely, well-defined 
and, in fact, Lipschitz. As in \cite{dashti,schwab} we may then define
a posterior $\nu^y$ on $u$ which has density with respect to $\nu_0$ 
given by \eqref{eq1}:
\begin{equation}
\label{eq:posterior}
\frac{d\nu^y}{d\nu_0}(u) \propto \pi(u)=\exp\bigl(-\Phi(u;y)\bigr).
\end{equation}
Exploring the posterior distribution $\nu^y$ is the objective of the paper.
In doing so, the following fact will be relevant; it is easily established
by using the fact that \eqref{eq:deen}
holds almost surely for $u \sim \nu_0$, together
with the bound on the solution of the elliptic PDE 
given in \eqref{eq:est}. 

\begin{lem}
\label{lem:note}
There is a constant $\pi_{\rm min}=\pi_{\rm min}(y)>0$ such that $\pi_{\rm min}
\le \pi(u) \le 1$ almost surely for $u \sim \nu_0.$
\end{lem}

We finish by noting that, in algorithmic practice, 
it is typically necessary (see, however, \cite{trunc} in the context of MCMC) to apply a spectral truncation:
\begin{equation}
u(x) = \bar{u}(x) + \sum_{k:\, |k|_{\infty}<c} u_k e_k(x) 
\label{truncation}
\end{equation}
where $c$ is a truncation parameter.
Having defined the desired parameterization of $u$, we consider the truncated 
vector of Fourier coefficients $u_k$ as the object of inference in practice.

\section{Sequential Monte Carlo}
\label{sec:SMC}

In this section we describe the application of SMC methods to
Bayesian inversion. Subsection \ref{ssec:smc} contains an
explanation of the basic methodology and statement of
a basic (non-adaptive) algorithm. Subsection \ref{ssec:t}
contains statement and proof of a convergence theorem for the
basic form of the algorithm, notable because it applies
in infinite dimensional spaces. In subsection \ref{ssec:a}
we describe an adaptive version of the SMC algorithm, which we
use in practice.

\subsection{Standard SMC Samplers}
\label{ssec:smc}

Let $(E,\mathscr{E})$ denote a measure space and $\nu_0$ a probability
measure on that space. We wish to sample from a target probability 
measure $\nu$ on $(E,\mathscr{E})$, which has density with respect to
$\nu_0$ known up to a normalizing constant:
\begin{equation}
\label{eq:posterior}
\frac{d\nu}{d\nu_0}(u) \propto \pi(u).
\end{equation}
We introduce a sequence of ``bridging'' densities which enable us to
connect $\nu_0$ to $\nu$: 
\begin{equation}
\label{eq:aux}
\pi_n(u) \propto \pi(u)^{\phi_n}\ , \quad x\in E,
\end{equation}
where
$
0=\phi_0<\cdots < \phi_{n-1}< \phi_n < \cdots<\phi_p=1; 
$
we refer to the $\phi_j$ as \emph{temperatures}.
We let $\nu_n$ denote the probability measure with density proportional
to $\pi_n$ with respect to $\nu_0$. Assuming that $\pi(u)$ is finite 
$\nu_0$ almost surely we obtain 
\begin{equation}
\label{eq:posteriorn}
\frac{d\nu_n}{d\nu_{0}}(u) \propto \pi(u)^{\phi_n},\quad\quad
\frac{d\nu_n}{d\nu_{n-1}}(u) \propto \ell_{n-1}(u):=\pi(u)^{\phi_n-\phi_{n-1}},
\quad n\in\{1, \cdots, p\}.
\end{equation}
We note that the assumption on $\pi$ being finite is satisfied
for our elliptic inverse problem; see Lemma \ref{lem:note}.
Although $\nu=\nu_p$ may be far from $\nu_0$, careful choice
of the $\phi_n$ can ensure that $\nu_{n}$ is close to $\nu_{n-1}$
allowing gradual evolution of approximation of $\nu_0$ into
approximation of $\nu$. Other choices of 
bridging densities are possible and are discussed in e.g.~\cite{delm:06}.

Let $\{L_{n-1}\}_{n=1}^p$ denote the sequence of (nonlinear) maps on measures
found by applying Bayes's Theorem with likelihood proportional to
$\{\ell_{n-1}\}_{n=1}^p$ and let $\{K_{n}\}_{n=1}^p$ 
be a sequence of Markov kernels (and equivalently, for notational convenience,
the resulting linear maps on measures) with invariant measure 
$\{\nu_{n}\}_{n=1}^p$. 
We define $\{\Phi_n\}_{n=1}^p$ to be the nonlinear maps on measures
found as $\Phi_n=K_n L_{n-1}$.  Explicitly we have, for each $n\geq 1$ and any 
probability measure $\mu$ on $E$:
$$
(\Phi_n \mu)(du) = \frac{\mu(\ell_{n-1}K_n(du))}{\mu(\ell_{n-1})}
$$
where we use the notation 
$\mu(\ell_{n-1}K_n(du)) = \int_E \mu(dy) \ell_{n-1}(y) K_n(y,du)$ 
and $\mu(\ell_{n-1}) = \int_E \mu(dy) \ell_{n-1}(y)$. It then follows that
\begin{equation}
\label{eq:F}
\nu_n=\Phi_n\nu_{n-1}, \quad n\in\{1, \cdots, p\}.
\end{equation}
The standard SMC algorithm is described in Figure \ref{tab:SMC}. It involves
a population of $M$ particles evolving with $n$.
With no resampling, the algorithm coincides with 
annealed importance sampling as in \cite{neal:01}. With resampling at every step
(i.e.\@ the case $M_{\rm thresh}=M$, where $M_{\rm thresh}$ denotes the cut-off point 
for the Effective Sample Size (EES)) we define the empirical approximating measures by
the iteration
\begin{equation}
\label{eq:PF}
\nu_n^M= S^M \Phi_n \nu_{n-1}^M, \quad n\in\{1, \cdots, p\}; \quad\quad \nu_0^M=\frac{1}{M}\sum_{m=1}^M\delta_{u_0^m}
\end{equation}
Here
$$(S^M \mu)(dv) = \frac{1}{M} \sum_{m=1}^M \delta_{v^{(m)}}(dv), \quad v^{(m)} \sim \mu ~~~ {\rm i.i.d.}.$$


\begin{figure}[!h]
\begin{flushleft}
\medskip
\hrule
\medskip
{\itshape
\begin{enumerate}
\item[\textit{0.}] Sample $\{u_0^m\}_{m=1}^M$ i.i.d.\@ from $\nu_0$ and define
the weights $w_0^m=M^{-1}$ for $m=1,\cdots, M.$ 
%
%
Set $n=1$ and $l=0$.
\vspace{0.18cm}
\item[\textit{1.}]  For each $m$ set $\hat{w}_n^m=\ell_{n-1}(u_{n-1}^m)w_{n-1}^m$
and sample $u_n^m$ from $K_n(u_{n-1}^m,\cdot);$ calculate the normalized weights
\begin{equation*}
w_n^m  = \hat{w}_n^m/\bigl(\sum_{m=1}^M \hat{w}_n^m\bigr).
\end{equation*}
%
\item[\textit{2.}] Calculate the Effective Sample Size (ESS):
\begin{equation}
\label{eq:ess_def}
\textrm{ESS}_{(n)}(M) := \frac{\left(\sum_{m=1}^M w_n^m\right)^2}{\sum_{m=1}^M (w_n^m)^2} \ .
\end{equation}
If $ESS_{(n)}(M) \le M_{\rm thres}$: \\
\hspace{0.3cm} resample $\{u_{n}^{m}\}_{m=1}^M$ 
according to the normalized weights $\{w_n^m\}_{m=1}^M$;\\ 
\hspace{0.3cm} re-initialise the weights by setting $w_{n}^m=M^{-1}$
for $m=1,\cdots, M$;\\
\hspace{0.3cm} let $\{u_{n}^{m}\}_{m=1}^M$ 
now denote the resampled particles.\\

\item[\textit{3.}] If $n<p$ set $n=n+1$ and return to Step 1; otherwise stop.

%
\end{enumerate} }
\medskip
\hrule
\medskip
\end{flushleft}
\vspace{-0.5cm}
\caption{Standard SMC Samplers. $M_{\rm thres}\in\{1,\ldots, M\}$ 
is a user defined parameter.}
\label{tab:SMC}
\end{figure}

\subsection{Convergence Property}
\label{ssec:t}

The issue of dimensionality in SMC methods has attracted substantial attention in the literature \cite{beskos,beskos1,beskos2,rebs}. In this section, using 
a simple approach for the analysis of particle filters which is clearly 
exposed in \cite{rebs}, we show that for
our SMC method it is possible to prove dimension-free error bounds. Whilst the 
theoretical result in this subsection is not entirely new (similar results 
follow from the work in \cite{delmoral,delmoral1}), a direct and simple proof 
is included for non-specialists in SMC methods,
to highight the utility of SMC methods in inverse problems, and to connect
with related recent results in dimension-independent MCMC, such as
\cite{HSV}, which are far harder to establish.

We will consider the algorithm in Figure \ref{tab:SMC} with $M_{\rm thres}=M$, 
so one resamples at every time step (and this is multinomially). 
Note that then, for $n\geq 0$, at the end of each step of
the algorithm the approximation to $\nu_n$ is given by
\begin{equation}
\label{eq:A}
\nu_n^M(du) := \frac{1}{M}\sum_{m=1}^M \delta_{u_n^m}(dx), 
\end{equation}
which follows from the algorithm in Figure \ref{tab:SMC} with $M_{\rm thres}=M$
or, equivalently, \eqref{eq:PF}.

Throughout, we will assume that there exists a $\kappa>0$ such that for each $n\geq 0$ and any $u\in E$
\begin{equation}
\kappa \leq \ell_n(u) \leq 1/\kappa\label{eq:hyp}.
\end{equation}
We note that this holds for the elliptic inverse problem from the previous
section, when the uniform prior $\nu_0$ is employed; see Lemma \ref{lem:note}.

Let $\cP$ denote the collection of all probability measures on $E$. 
Let $\mu=\mu(\omega)$ and $\nu=\nu(\omega)$ denote two possibly random 
elements in $\cP$,  and $\mathbb{E}^{\omega}$ expectation w.r.t.~$\omega.$ 
We define the distance between $\mu, \nu \in \cP$ by 
$$
d(\mu,\nu) = {\rm sup}_{|f|_\infty \leq 1} \sqrt{ \mathbb{E}^{\omega} |\mu(f)-\nu(f)|^2 },
$$
where the supremum is over all $f: E \to \bbR$ with 
$|f|_\infty:=\sup_{v\in E}|f(v)| \leq 1.$
This definition of distance is indeed a metric on the space of random
probability measures; in particular it satisfies the triangle inequality.
In the context of SMC the randomness underlying the approximations
\eqref{eq:A} comes from the various sampling operations within the
algorithm.

We have the following convergence result for the SMC algorithm.

\begin{theorem}\label{th39}
Assume \eqref{eq:hyp} and consider the SMC algorithm with
$M_{\rm thresh}=M$. Then, for any $n\geq 0$,
$$
d(\nu_n^M,\nu_n) \le \sum_{j=0}^n (2\kappa^{-2})^j \frac{1}{\sqrt M}.
$$
\end{theorem}

\begin{proof}
For $n=0$ the result holds via Lemma \ref{lem:tech_lem2}. For $n>0$, 
we have, by the triangle inequality, Lemma \ref{lem:tech_lem} and
Lemma \ref{lem:tech_lem2} (which may be used by the conditional
independence structure of the algorithm), 
\begin{align*}
d(\nu_n,\nu_n^M) & = d(\Phi_n \nu_{n-1}, S^M\Phi_n\nu_{n-1}^M)\\
& \le d(\Phi_n \nu_{n-1},\Phi_n \nu_{n-1}^M)+d(\Phi_n \nu_{n-1}^M,S^M\Phi_n\nu_{n-1}^M)\\
& \le \frac{2}{\kappa^2}d(\nu_{n-1},\nu_{n-1}^M)+\frac{1}{\sqrt{M}}.
\end{align*}
Iterating gives the desired result.
\end{proof}

\begin{rem}
\label{rem:fil}
This theorem shows that the sequential particle filter actually
reproduces the true posterior distribution $\nu_p$, 
in the limit $M \to \infty$. 
We make some comments about this.

\begin{itemize}

\item
The measure $\nu_p$ is well-approximated by 
$\nu_p^M$ in the sense that, as the number of particles
$M \to \infty$, the approximating measure converges to the
true measure. The result holds in the infinite dimensional
setting. As a consequence the algorithm as stated is
robust to finite dimensional approximation.


\item In principle the theory applies even if the Markov
kernel $K_n$ is simply the identity mapping on probability
measures. However, moving the particles according to a 
non-trivial $\nu_n$-invariant measure is absolutely 
essential for the methodology
to work in practice. This can be seen by noting that if $K_n$
is indeed taken to be the identity map on measures then the
particle positions will be unchanged as $n$ changes, meaning
that the measure $\nu_p$ is approximated by weighted
samples (almost) from the prior, clearly undesirable in general.

\item The MCMC methods in \cite{cotter} provide explicit examples
of Markov kernels with the desired property of preserving
the measures $\nu_n$, including the infinite dimensional setting.

\item In fact, if the Markov kernel $K_n$ has some ergodicity properties 
then it is sometimes possible to obtain bounds which are {\em uniform} in $p$; see \cite{delmoral,delmoral1}.

\end{itemize}
\end{rem}

\begin{lem}\label{lem:tech_lem}
Assume \eqref{eq:hyp}. Then, for any $n\geq 1$ and any $\mu,\nu \in \cP$,
$$
d(\Phi_n\mu,\Phi_n\nu) \leq \frac{2}{\kappa^2}d(\mu,\nu).
$$
\end{lem}
\begin{proof}
For any measurable $f:E\rightarrow\mathbb{R}$ 
we have
\begin{align*}
[\Phi_n\mu-\Phi_n\nu](f)  &= \frac{1}{\mu(\ell_{n-1})}[\mu-\nu](\ell_{n-1}K_n(f))\\
&\quad\quad\quad\quad\quad + 
\frac{\nu(\ell_{n-1}K_n(f))}{\mu(\ell_{n-1})\nu(\ell_{n-1})} [\nu-\mu](\ell_{n-1}).
\end{align*}
So we have, by Minkowski,
\begin{align*}
\mathbb{E}^{\omega}[|[\Phi_n\mu-\Phi_n\nu](f)|^2]^{1/2} & \leq
\mathbb{E}^{\omega}\Big[\Big|\frac{1}{\mu(\ell_{n-1})}[\mu-\nu](\ell_{n-1}K_n(f))\Big|^2\Big]^{1/2}\\
&
\quad\quad\quad\quad\quad+ \mathbb{E}^{\omega}\Big[\Big|\frac{\nu(\ell_{n-1}K_n(f))}{\mu(\ell_{n-1})\nu(\ell_{n-1})} [\nu-\mu](\ell_{n-1})]\Big|^2\Big]^{1/2}.
\end{align*}
Note that the ratio
$$\frac{\nu(\ell_{n-1}K_n(f))}{\nu(\ell_{n-1})}$$
is an expectation of $f$ and is hence bounded by $1$ in modulus,
if $|f|_{\infty} \le 1.$
Then using the fact that $|\ell_{n-1}K_n(f)|_{\infty}\leq \kappa^{-1}$ and $\ell_{n-1}\geq \kappa$  (see \eqref{eq:hyp}) we deduce that
\begin{align*}
\mathbb{E}^{\omega}[|[\Phi_n\mu-\Phi_n\nu](f)|^2]^{1/2} &\leq \frac{1}{\kappa^2}\mathbb{E}^{\omega}\Big[\Big|[\mu-\nu](\ell_{n-1}K_n(f)\kappa)\Big|^2\Big]^{1/2}\\
&\quad\quad\quad\quad\quad+ 
\frac{1}{\kappa^2}\mathbb{E}^{\omega}\Big[\Big|[\nu-\mu](\ell_{n-1}\kappa)]\Big|^2\Big]^{1/2}.
\end{align*}
using the fact that $|\ell_{n-1}K_n(f)|_{\infty}\leq \kappa^{-1}$ and $\ell_{n-1}\leq \kappa^{-1}$, with the first following from  \eqref{eq:hyp} together
with the Markov property for $K_n$,
taking suprema over $f$ completes the proof.
\end{proof}

\begin{lem}
\label{lem:tech_lem2}
The sampling operator satisfies
$$\sup_{\mu \in \cP}d(S^M \mu,\mu) \leq \frac{1}{\sqrt{M}}.$$
\end{lem}

\begin{proof}
Let $\nu$ be an element of $\cP(X)$ and $\{v^{(k)}\}_{k=1}^M$
a set of  i.i.d.\ samples with $v^{(1)} \sim \nu$; the randomness
entering the probability measures is through these samples, expectation
with respect to which we denote by $\bbE^{\omega}$ in what follows. Then
$$S^M \nu(f)=\frac{1}{M}\sum_{k=1}^M f(v^{(k)})$$
and, defining $\bF=f-\nu(f)$, we deduce that
$$S^M\nu(f)-\nu(f)=\frac{1}{M}\sum_{k=1}^M \bF(v^{(k)}).$$
It is straightforward to see that
$$\bbE^{\omega} \bF(v^{(k)})\bF(v^{(l)})=\delta_{kl}\bbE^{\omega} |\bF(v^{(k)})|^2.$$
Furthermore, for $|f|_\infty \le 1$,
$$\bbE^{\omega} |\bF(v^{(1)})|^2=\bbE^{\omega} |f(v^{(1)})|^2-|\bbE^{\omega} f(v^{(1)})|^2 \le 1.$$
It follows that, for $|f|_\infty \le 1$,
$$\bbE^{\omega}|\nu(f)- S^M\nu(f)|^2=\frac{1}{M^2}\sum_{k=1}^M \bbE^{\omega} |\bF(v^{(k)})|^2 \le \frac{1}{M}.$$
Since the result is independent of $\nu$ we may take the supremum over all
probability measures and obtain the desired result.
\end{proof}

\subsection{Adaptive SMC Samplers}
\label{ssec:a}

In practice, the SMC samplers algorithm requires the specification of $0<\phi_0<\cdots < \phi_{n-1}< \phi_n < \cdots<\phi_p=1$
as well as any parameters in the MCMC kernels. As demonstrated in \cite{jasra, kantas}, the theoretical validity of which is established in \cite{adaptive}, these parameters may be set on the fly.


First, we focus on the specification of the sequence of distributions.
Given step $n-1$ and $\pi_{n-1}(x)$, 
we select the next target density by adapting the temperatures to
a required value for the Effective Sample Size (ESS) statistic \eqref{eq:ess_def} as  in \cite{jasra} (see also \cite{zhou} for an alternative procedure).
%
So, for a user-specified threshold $M_{\textrm{thres}}$, we choose
$\phi_{n}$ as the solution of $ESS_{(n)}(M)=M_{\textrm{thres}}$. 
One can use an inexpensive bisection method to obtain $\phi_n$.

Second, we turn to the specification of the mutation kernels $K_n$.
Several options are available here, but we will use reflective random walk Metropolis proposals on each univariate component, conditionally independently.
We will adapt the random-move proposal scales, $\epsilon_{j,n}$, with $j$ the co-ordinate and $n$ the time index. A simple choice would be to tune $\epsilon_{j,n}$ to the marginal variance along 
the $j$-th co-ordinate; since this is analytically unavailable we opt 
for the SMC estimate at the previous time-step. Thus, we  
set $\epsilon_{n,j}=\rho_{n} \sqrt{\widehat{\mathrm{Var}}(U_{n,j})}$ 
where $\rho_n$ is a global scaling parameter. For $\rho_n$ itself, 
we propose to modify it
based on the previous average acceptance rate over the population of particles
(denoted $\alpha^N_{n-1}$),
to try to have average acceptance rates in a neighbourhood of 0.2 (see e.g.~\cite{beskos_non} and the references therein for a justification). Our adaptive strategy works as follows;
\begin{equation}
\rho_n = \left\{
\begin{array}{cl}
   2 \rho_{n-1}\ ,  &  \text{if } \alpha^N_{n-1} > 0.3 \\
   0.5 \rho_{n-1}\ , & \text{if } \alpha^N_{n-1} < 0.15 \\
   \rho_{n-1}\ , & \text{o.w.}
\end{array} 
\right.
\label{adapt_rho}
\end{equation}
Thus, we scale
$\rho$ upwards (downwards) if the last average acceptance rate went above (below) a predetermined neighbourhood of 0.2. This approach is different to the work in \cite{kantas}.

In addition, one can synthesize a number, say $M_n$, of baseline MCMC kernels, to obtain an overall effective one with good mixing; this is a new contribution relative to \cite{kantas}.
To adapt $M_n$, we follow the following heuristic;
We propose to select $M_n$ using $M_n = \lfloor \frac{m}{\rho_n^2} \rfloor$, with $m$ being a global parameter. The intuition is that for random-walk-type 
transitions of increment with small standard deviation $\delta$, one 
needs $\mathcal{O}(\delta^{-2})$ steps to travel distance $\mathcal{O}(1)$ in 
the state-space.
A final modification for practical computational reasons is that we force $M_n$ steps to lie within a predetermined bounded set, i.e. $[l,u]$.

The adaptive SMC algorithm works as in Figure \ref{tab:SMC}, except in step 1, before simulation from $K_n$ is undertaken, our adaptive procedure is implemented. Then one may resample (or not) and then move the samples according to $K_n$. In addition, the algorithm will
run for a random number of time steps and terminate when $\varphi_n=1$ (which will happen in a finite number of steps almost surely).
 
%

\section{Numerical Results}\label{sec:numerics}

In this section, we describe the details of our implementation (Section \ref{ssec:i}), describe the  objects of inference (Section \ref{ssec:oii})
and give our results in 2D (Section \ref{ssec:2d}) and 3D (Section \ref{ssec:3d}).


\subsection{Implementation Details}
\label{ssec:i}

The software used in our experiments has been implemented in C++ for the GNU$\backslash$Linux platform. We used the Libmesh library
for finite elements computation \cite{libMeshPaper}, we used 
the Fast Fourier Transform for rapid evaluation of the sum in
$u(\cdot)$ at pre-determined grid-points in ${\mathcal X}$
and we exploited parallel computation wherever possible, 
for which we used the MPI libraries.
Our experiments were run on a computer server with 23 ``Intel(R) Xeon(R)CPU X7460 @2.66GHz'' processors, each with 2 cores; 50 Gb memory and running ``RedHat Linux version 2.6.18-194.el5'' operating system. The experiments discussed in this paper used 20 processors.

All the colour plots of random fields (e.g. permeability fields) have been prepared using the rainbow color scheme from the R programming language/environment. The scheme quantizes the Hue quantity of HSV (Hue Saturation Value) triplet of a pixel. Our level of quantization is selected to be 256 (8 bits), with the Hue range of $[0,1]$, hence we normalize the random fields to this range and quantize to 8 bits to get the Hue value for a pixel. Saturation and Value were taken to be 1. All images were computed using $500 \times 500$ equi-spaced point evaluations from the respective random fields.



\subsection{Objects of Inference}
\label{ssec:oii}

The work in \cite{vollmer} investigates the performance of the Bayesian approach for our elliptic inverse problem and
gives sufficient conditions under which posterior consistency holds. Posterior consistency is concerned with ``ball probabilities''
of type
\[\lim_{{\textrm{Card}(O) \to \infty}} 
\int_{B_{\epsilon}}\frac{d\nu^y}{d\nu_0}(u)\nu_0(du) = 1
\]
where $y=\{y_x\}_{x \in O}$ and
$B_{\epsilon}$ is the $\epsilon$ neighbourhood of the true value of $u$.
One way to check such a result numerically is to use the posterior estimates obtained via our method. The estimated ball probabilities are computed as follows:
\begin{equation}
\sum_i w_p^i\mathbb{I}_{B_{\epsilon}}(u(x_p^i))
\label{ball-estimate}
\end{equation}

Although not all the conditions in \cite{vollmer} required for posterior 
consistency to hold are fulfilled, 
we will nonetheless empirically investigate such a consistency property.
This also provides a severe test for the SMC method since 
it implies posterior measures in the large dataset limit.

\subsection{Results in 2D}
\label{ssec:2d}

We consider equation \eqref{eq2} in dimension $d=2$ and with
source and sinks as specified in \eqref{eq3}. Our goal is to construct a sequence of posterior estimates,
corresponding to increasing number of observations in order to numerically illustrate posterior consistency. Table \ref{tab_2d_param}
shows the parameters used in our experiments.

\begin{table}[!htb]
\begin{tabular}{|l|c|}
  \hline
  {\bf Parameter name} & {\bf Value} \\ \hline
  frequency cutoff & 10 \\ \hline
  Finite Elements d.o.f.  & 100 \\ \hline
$\sigma^2$ & $5 \times 10^{-7}$ \\ \hline
  $M$ & 1000 \\ \hline
  $M_{\rm thres}$ & 600 \\ \hline
  $a$   & 4 \\ \hline
  $\bar{u}$ & 40 \\ \hline
  Wall-clock time & 11 hrs \\ \hline
\end{tabular}
\centering
\caption{Parameter values of used for the 2D experiments. Between 5 and 1000 steps are allowed for the iterates of the MCMC kernels.
The frequency cutoff determines the level of discretization of the permeability field. Finite
elements d.o.f. denotes the number of finite elements used in the numerical solution of the elliptic PDE, higher values indicate better approximation
at the expense of computational resources. For $a$ see \eqref{eq:adef}.}
\label{tab_2d_param}
\end{table}

To get an empirical sense of these parameters' effects on the distribution of the permeability field, we plot some samples from the prior field $u(x)$
in Figure \ref{prior_perm}. We will then generate 100 data points from the model, in the scenario where the permeability field is as Figure \ref{true_plot2D} (a).
In Figure \ref{true_plot2D} (b) we show the mean of the posterior permeability field, 
using these $100$ noisily observed data points. The posterior mean is estimated
as the mean of the particle distribution from SMC at the final time,
and notable in that it shows good agreement with true value of the 
permeability.


%

\begin{figure}[!htb]
        \center{\includegraphics[scale=0.21] {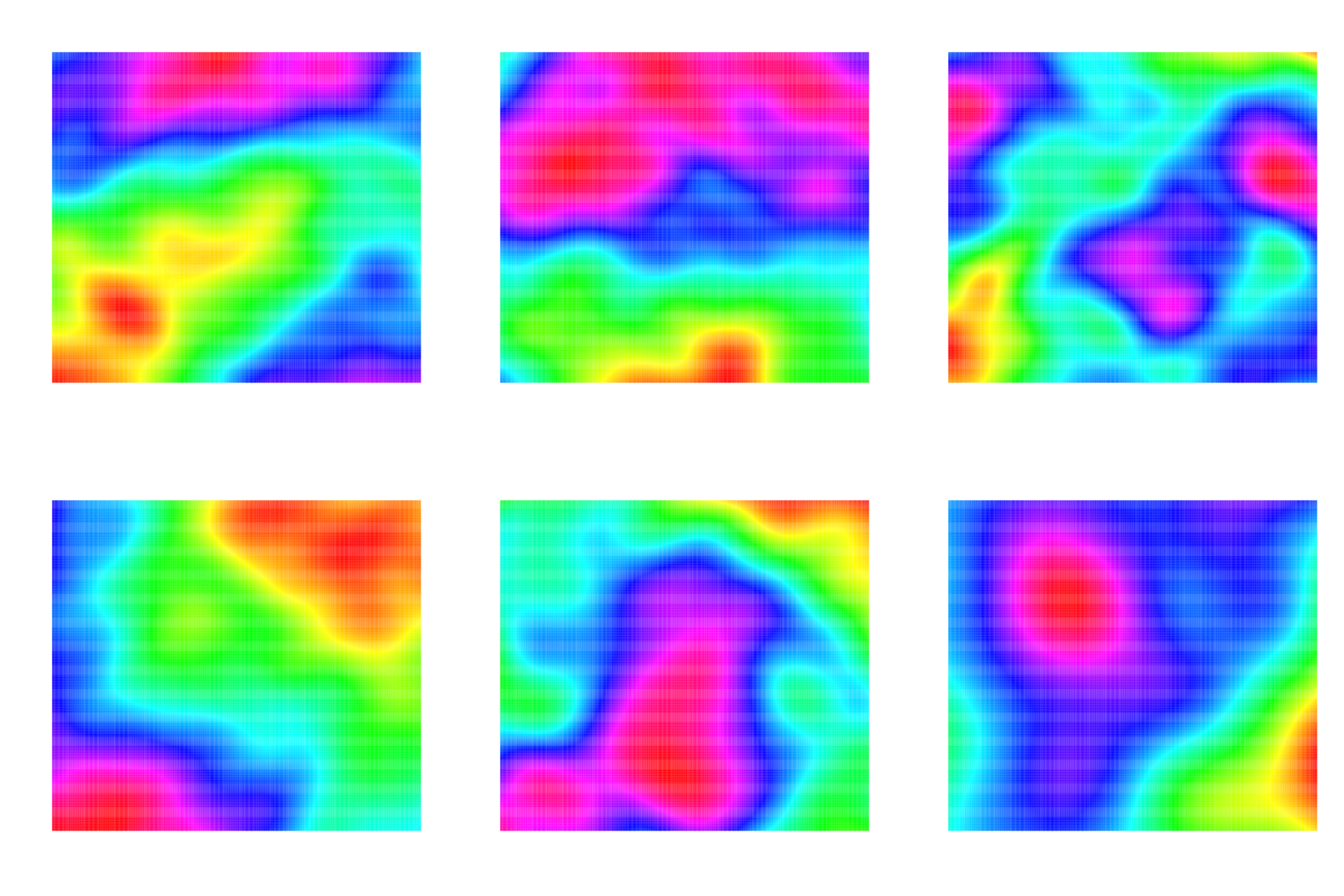}}
        \caption{Six permeability field samples drawn from the prior}
\label{prior_perm}
\end{figure} 

\begin{figure}
 \begin{subfigure}[t]{2.9in}
        \centering
        \includegraphics[height=2.9in]{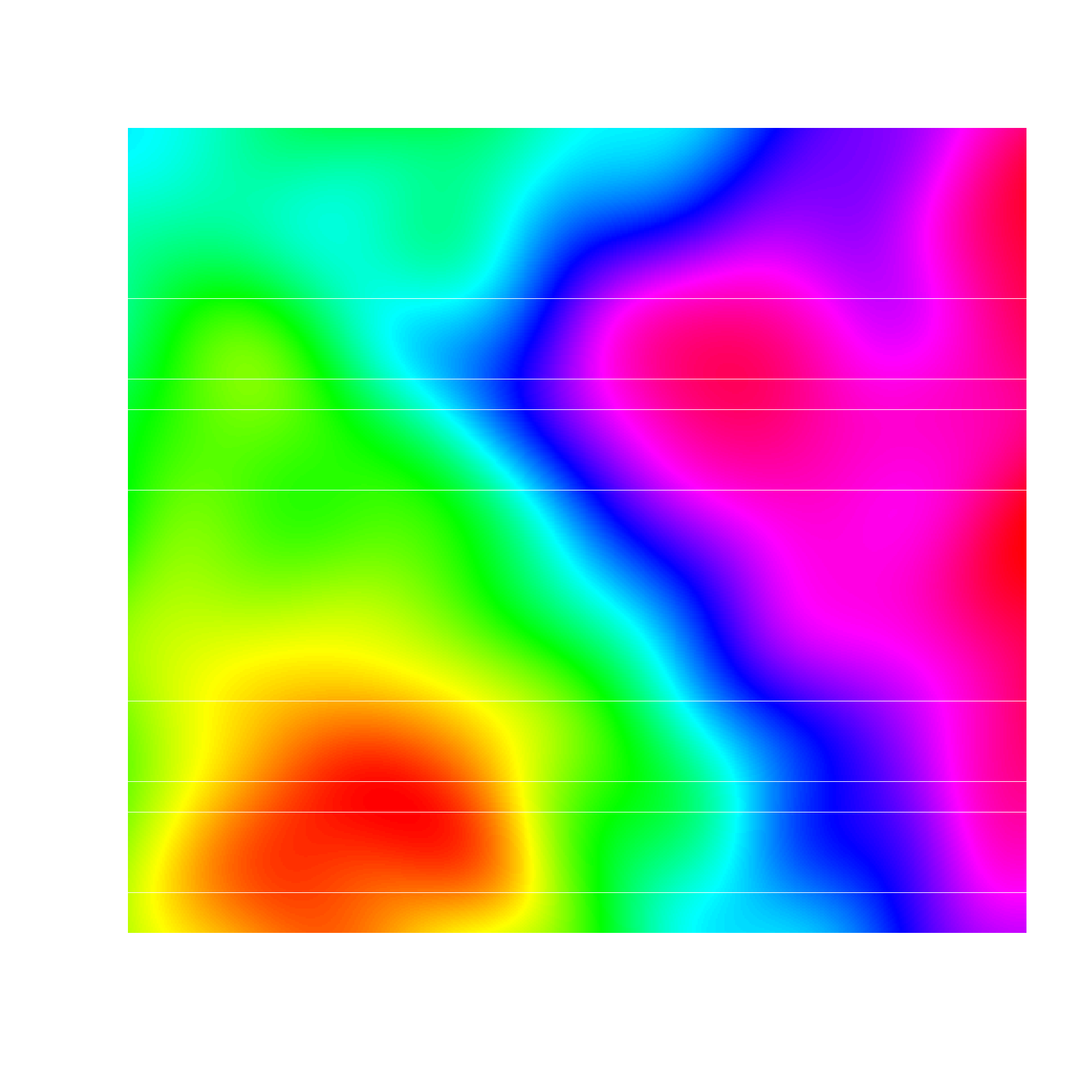}
        \caption{True permeability field.}
    \end{subfigure}
 \begin{subfigure}[t]{2.9in}
\includegraphics[height=2.9in]{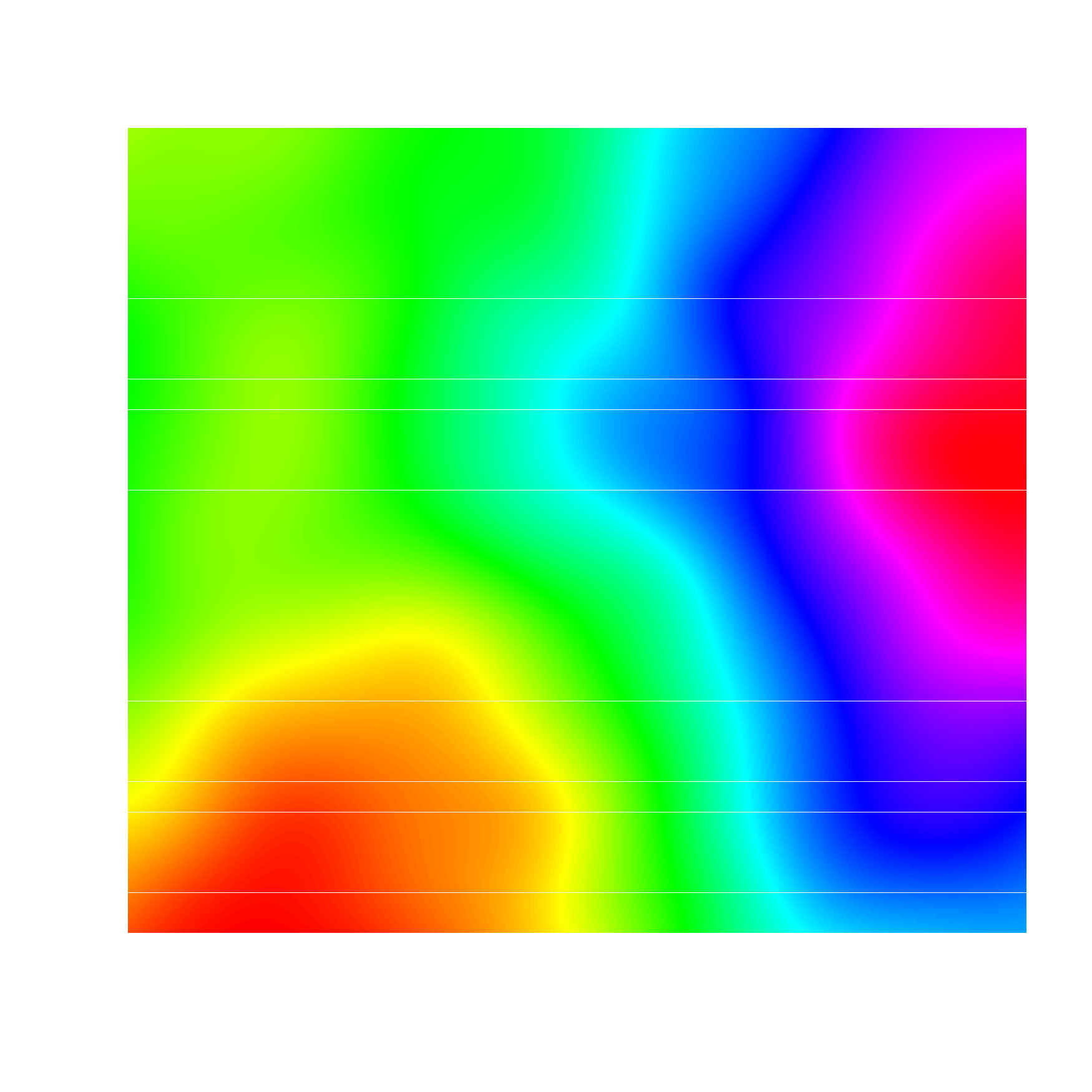}
        \caption{Estimated permeability field.}
    \end{subfigure}
\caption{The true and Posterior Estimated Permeability Field. The estimated filed is the mean estimated at the final step of the SMC method.
}
\label{true_plot2D}
\end{figure} 


In another experiment, designed to study posterior consistency,
a sequence of posterior estimates are formed by repeatedly running the adaptive SMC algorithm with, respectively, $4, 16, 36, 64$ and $100$ 
observations equi-spaced inside the domain of $[-\pi/2, \pi/2]^2$. The computed MSE and ball probabilities 
are given in Figure \ref{ball}, with the ball radius $\epsilon$ taken to be $0.17 \times 360$, where 360 is the number of parameters in the system, corresponding to a frequency cutoff of 10. 
The Figure suggests that as more data become available posterior
consistency is obtained as predicted, under slightly more restrictive 
assumptions than we have in play here, in \cite{vollmer}. This is interesting for
two reasons: firstly it suggests the potential for more refined Bayesian
posterior consistency analyses for nonlinear PDE inverse problems; 
secondly it demonstrates the potential to solve hard practical
Bayesian inverse problems and to obtain informed inference from a
relatively small number of observations.




\begin{figure}[!htb]
    \centering
    \begin{subfigure}[t]{2.9in}
        \centering
        \includegraphics[height=2.9in]{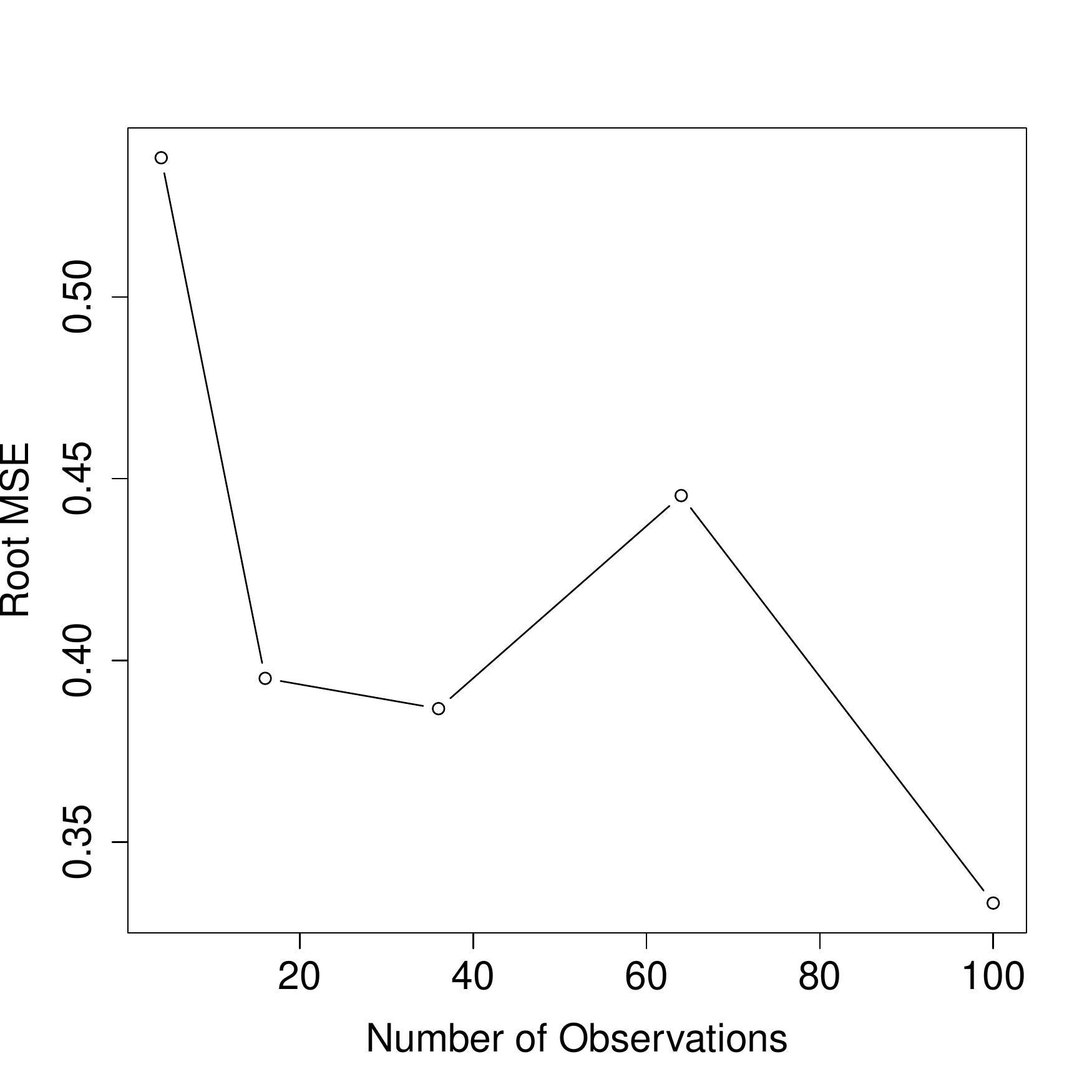}
        \caption{Root Mean Squared Error corresponding to the sequence of experiments}
    \end{subfigure}
    \begin{subfigure}[t]{2.5in}
        \centering
        \includegraphics[height=2.9in]{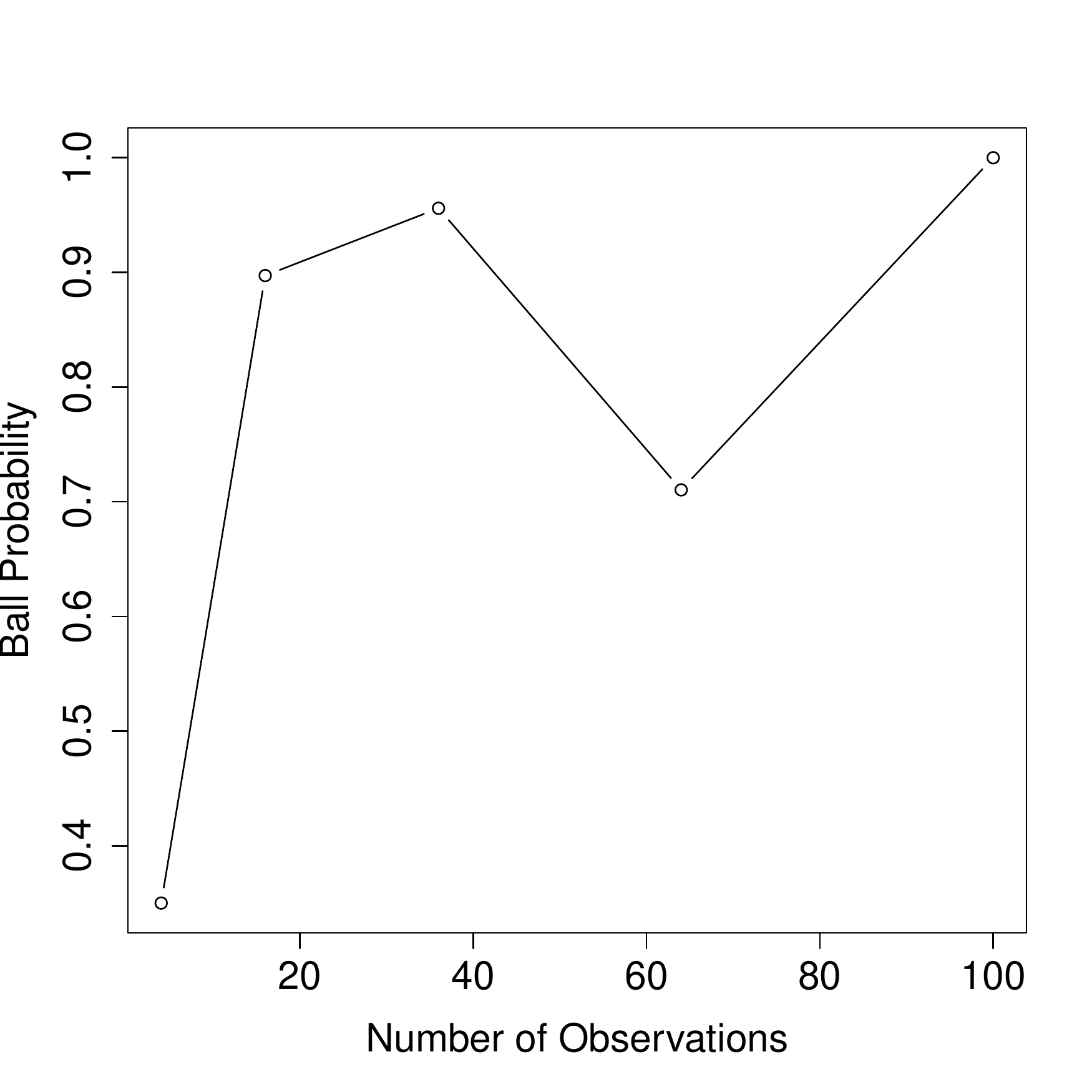}
        \caption{Estimated small-ball probabilities (Eqn. \ref{ball-estimate}) corresponding to the sequence of experiments}
    \end{subfigure}
    \caption{Numerical consistency checks for the sequence of experiments with 4,16,36,64 and 100 observations}
    \label{ball}
\end{figure}

Finally, Figure \ref{density2D} shows marginal posterior density estimates corresponding to 144 observations. The usual observation is to note the effectiveness of even the mode estimator in lower frequencies. Another important observation is the similarity of the high frequency marginal densities to the prior. In fact, it is this behaviour that makes a prior invariant MCMC proposal superior to others, i.e. the proposal itself is almost optimal for a wide range of coefficients in the problem.

\begin{figure}[!htb]
        \center{\includegraphics[scale=0.28] {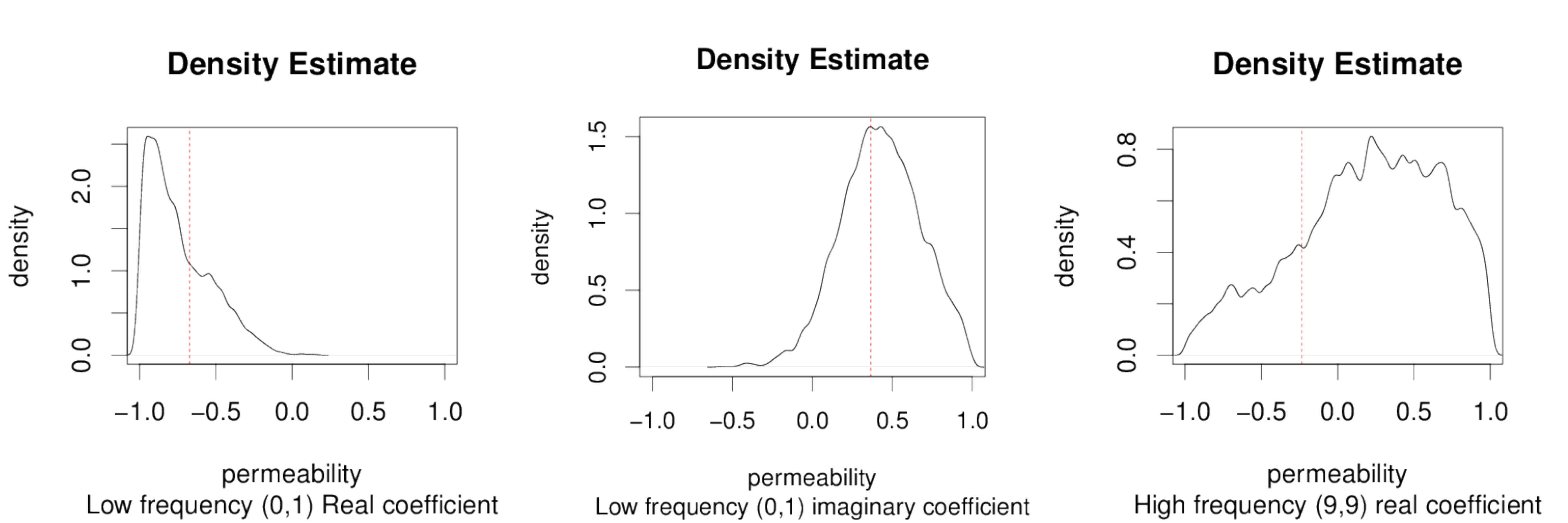}}
        \caption{Posterior marginal density estimates for two low and one high frequency coefficients in the 2D case}
\label{density2D}
\end{figure} 


\subsection{Results in 3D}
\label{ssec:3d}

A more realistic experiment is performed using the 3D setup discussed in Section \ref{3dinv}. In this setup, the computational 
aspects of the problem are further highlighted as the numerical solution of the forward operator becomes much harder due to the increased cardinality of the finite elements basis.
The values of parameters in this numerical study are given in Table \ref{3Dpar}. The data are generated from the model, under the specifications given in Table \ref{3Dpar}.

\begin{table}[!htb]
\begin{tabular}{|l|c|}
  \hline
  {\bf Parameter name} & {\bf Value} \\ \hline
  \# of Observations & 125 \\ \hline
  Frequency Cutoff & 5 \\ \hline
  Finite Elements d.o.f.  & 1000 \\ \hline
  $\sigma^2$ & $1 \times 10^{-8}$ \\ \hline
  $M$ & 1000 \\ \hline
  $M_{\rm thres}$ & 600 \\ \hline
  $a$   & 1 \\ \hline
  $\bar{u}$ & 100\\ \hline
  Wall-clock time & 10 days \\ \hline
\end{tabular}
\centering
\caption{Parameter values used for the 3D experiment discussed in this section. Between 5 and 200 steps are allowed for the iterates of the MCMC kernels.}
\label{3Dpar}
\end{table}

In Figure \ref{smc_performance}, we consider the performance of our SMC algorithm in this very challenging scenario. In Figure \ref{smc_performance} (a), we
can see the average acceptance rates of the MCMC moves, over the time parameter of the SMC algorithm. We can observe that these acceptance rates do not collapse to zero
and are not too far from 0.2. This indicates that the step-sizes are chosen quite reasonably by the adaptive SMC algorithm and the MCMC kernels have some mixing ability.
In Figure \ref{smc_performance} (b), we can see the number of MCMC iterations that are used per-particle over the time parameter of the SMC algorithm. We can observe,
as one might expect, that as the target distribution becomes more challenging, the number of MCMC steps required grows. Figure \ref{smc_performance} indicates reasonable
performance of our SMC algorithm.


In terms of inference, the posterior density estimates are shown in figure \ref{density}. Recall that the priors are uniform.
These estimates indicate a clear deviation from the prior specification, illustrating that the
data  influence our inference significantly. This is not obvious, and establishes that one can hope to use this Bayesian model in real applications.


\begin{figure}[!htb]
    \centering
    \begin{subfigure}[t]{2.9in}
        \centering
        \includegraphics[height=2.9in]{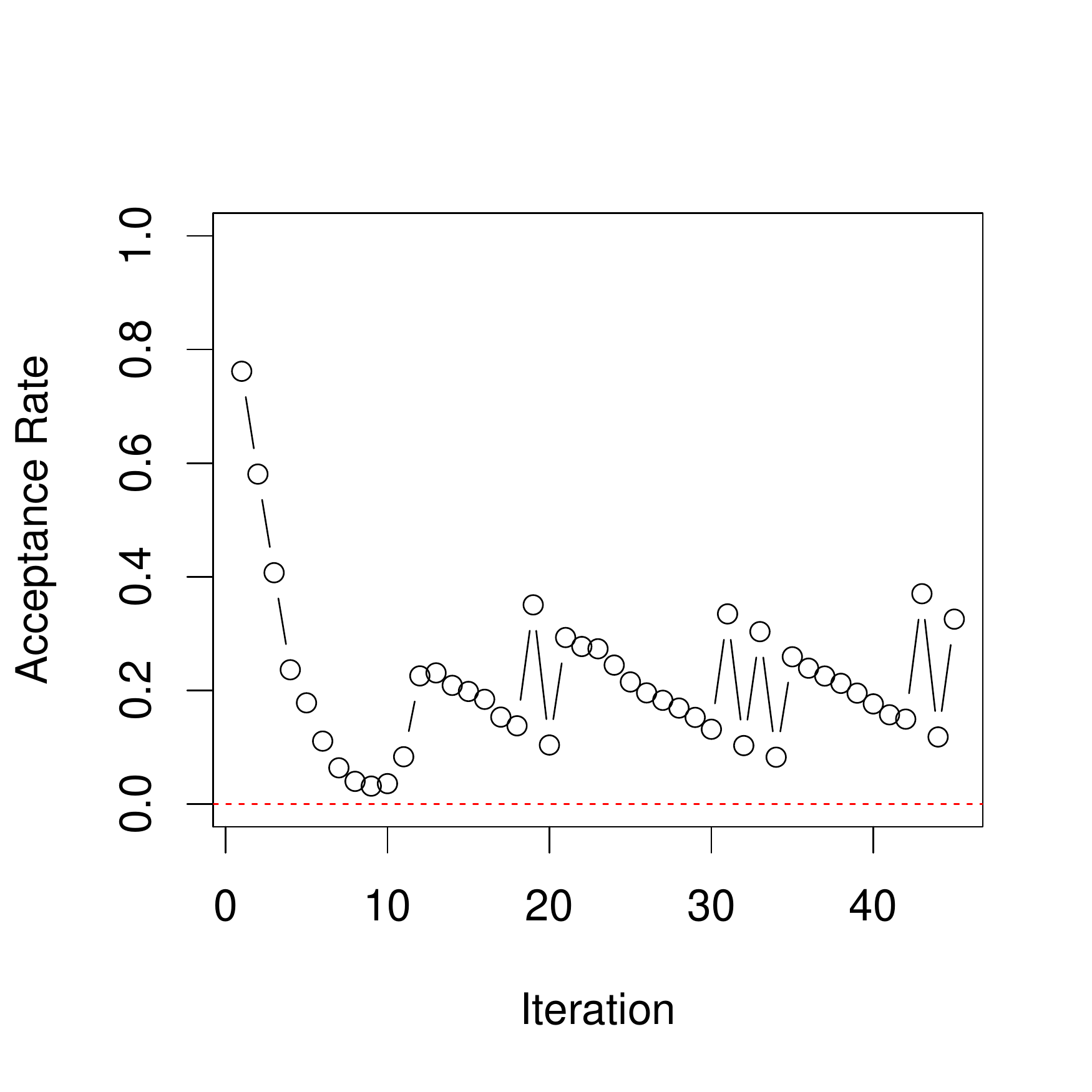}
        \caption{Average acceptance rates for each iteration of the adaptive-SMC algorithm}
    \end{subfigure}\,\,
    \begin{subfigure}[t]{2.5in}
        \centering
        \includegraphics[height=2.9in]{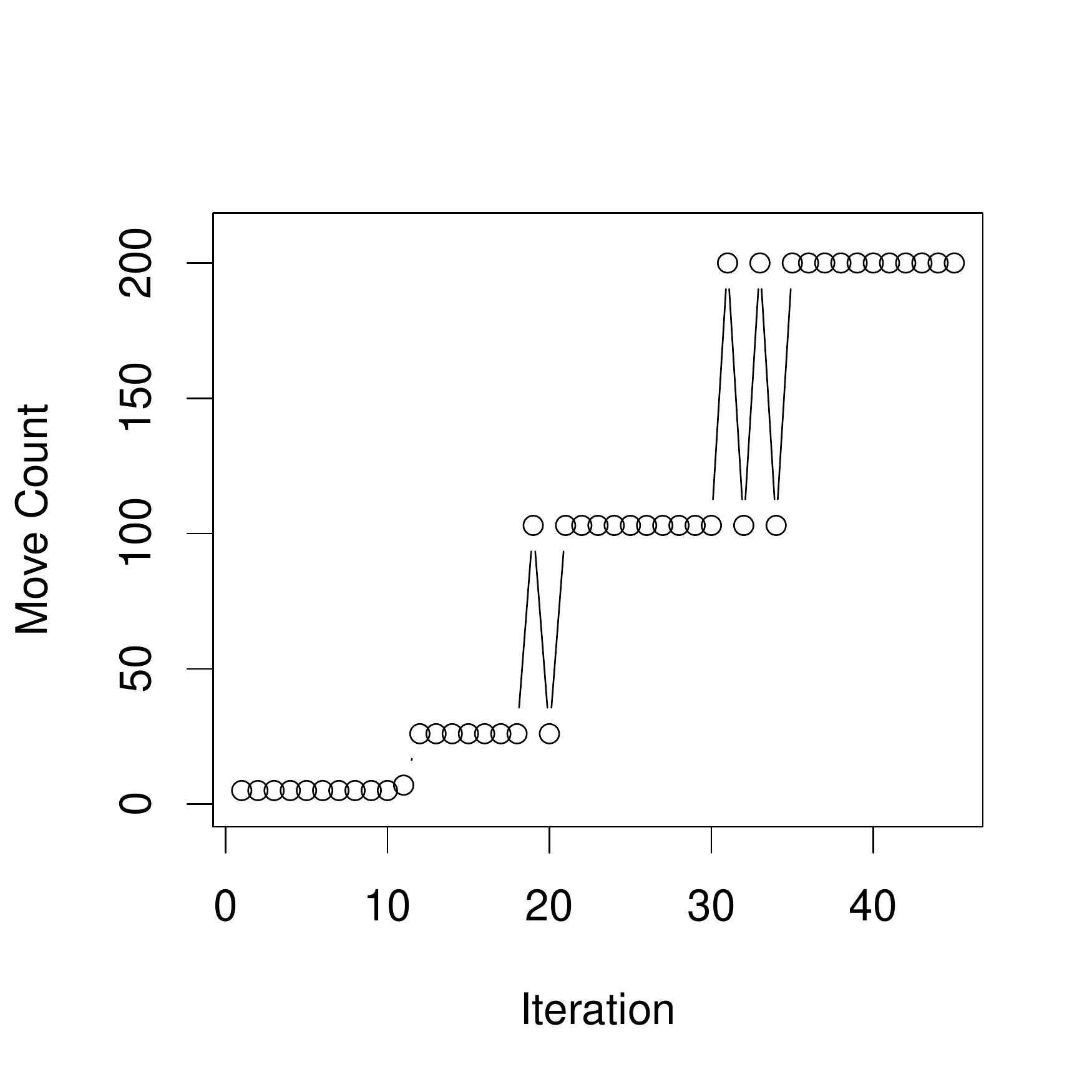}
        \caption{Adaptively selected MCMC steps for each iteration of the adaptive-SMC algorithm, bounded to be in [5,200].}
    \end{subfigure}
    \caption{SMC Performance for 3D Example.}
    \label{smc_performance}
\end{figure}

\begin{figure}[!htb]
        \center{\includegraphics[scale=0.28] {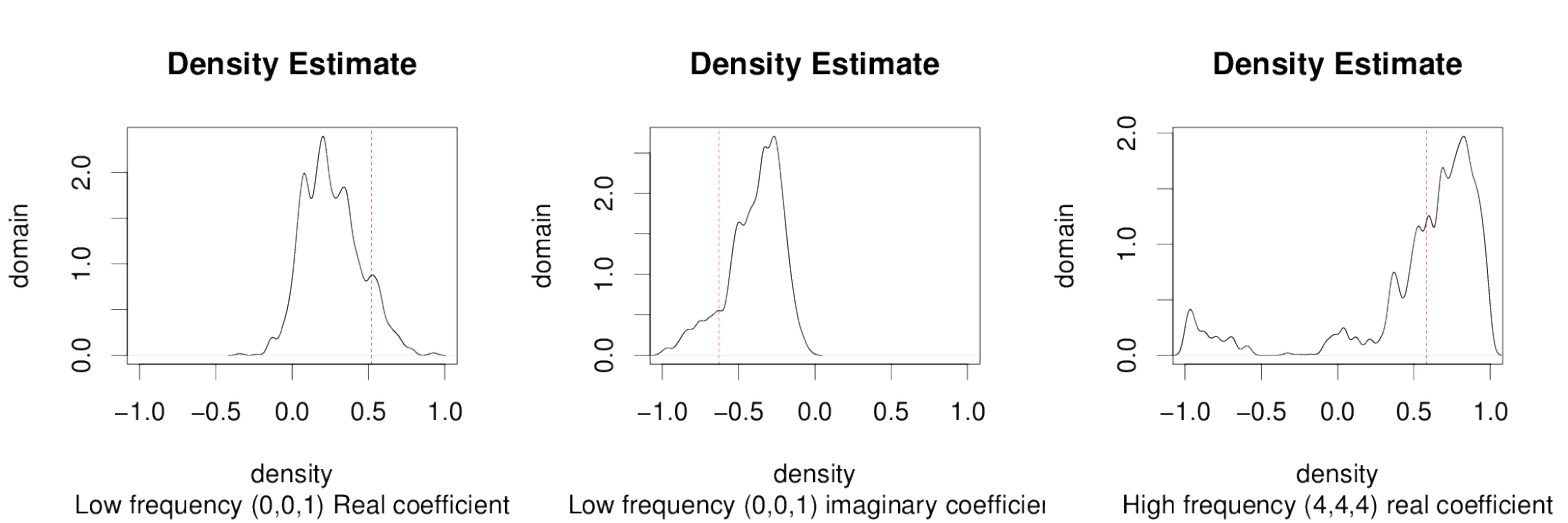}}
        \caption{Posterior marginal density estimates for two low and one high frequency coefficients in the 3D case}
\label{density}
\end{figure}

\section{Summary}\label{sec:summ}

In this article we have presented an SMC method for Bayesian inverse 
problems and applied it to a particular elliptic PDE inversion; the 
methodology, however, is transferable to other PDE inverse problems. 
Simulations demonstrated both the feasability of the SMC method for
challenging infinite dimensional inversion, as well as the property
of posterior contraction to the truth. In addition to simulations, we 
have provided a straightforward proof of the fact that SMC methods are 
robust to the dimension of the problem. 

There are several avenues for future research. Firstly, our error bounds 
explode w.r.t. the time parameter. It is of interest to find realistic 
conditions for which this is not the case (for instance the bounds in 
\cite{delmoral,delmoral1,whiteley} have assumptions which either do not hold or are hard to verify). 
Secondly, a further algorithmic innovation is to use multi-level Monte Carlo 
method as in \cite{hoang}, within the SMC context; this is being considered in  \cite{ml_smc}.
And finally it is of interest to consider the use of these methods to solve
other Bayesian inference problems.

\end{document}